\documentclass{article}
\usepackage[utf8]{inputenc}
\usepackage{amsmath,amsthm,amssymb}
\usepackage{mathtools}
\usepackage{algorithm}
\usepackage{todonotes}
\usepackage{xpatch}
\usepackage{algorithmic}
\usepackage{dsfont}
\usepackage{setspace}
\usepackage{bm}

\usepackage[preprint, nonatbib]{neurips_2020}
\usepackage[style=authoryear,backend=biber]{biblatex}
\renewbibmacro{in:}{}

\usepackage[utf8]{inputenc} 
\usepackage[T1]{fontenc}    
\usepackage{hyperref}       
\usepackage{url}            
\usepackage{booktabs}       
\usepackage{amsfonts}       
\usepackage{nicefrac}       
\usepackage{microtype}      

\newcommand{\pyear}[1]{\citeauthor{#1} (\citeyear{#1})}
\newcommand{\ah}[0]{\bm{\vec{a}}}
\newcommand{\norm}[1]{\left\lvert#1\right\rvert}

\newtheorem{theorem}{Theorem}[section]
\newtheorem{lemma}[theorem]{Lemma}
\newtheorem{prop}[theorem]{Proposition}
\theoremstyle{definition}
\newtheorem{corollary}{Corollary}[theorem]
\newtheorem{definition}{Definition}[section]

\title{A Simple Adaptive Procedure Converging to Forgiving Correlated Equilibria}

\author{%
  Hugh Zhang\thanks{Originally published in 2020 as a senior honors thesis under the supervision of Gabriel Carroll at https://purl.stanford.edu/hk596cg1085.} \\
  Department of Economics and Computer Science \\
  Stanford University \\
  \texttt{hughz@cs.stanford.edu}
}

\begin{document}

\maketitle

\begin{abstract}
Simple adaptive procedures that converge to correlated equilibria are known to exist for normal form games \parencite{hart2000simple}, but no such analogue exists for extensive-form games.
Leveraging inspiration from \citeauthor{zinkevich2008regret} (\citedate{zinkevich2008regret}), we show that any internal regret minimization procedure designed for normal-form games can be efficiently extended to finite extensive-form games of perfect recall.
Our procedure converges to the set of forgiving correlated equilibria, a refinement of various other proposed extensions of the correlated equilibrium solution concept to extensive-form games \parencite{forges1986approach, forges1986correlated, von2008extensive}. 
In a forgiving correlated equilibrium, players receive move recommendations only upon reaching the relevant information set instead of all at once at the beginning of the game.
Assuming all other players follow their recommendations, each player is incentivized to follow her recommendations \emph{regardless of whether she has done so at previous infosets}.
The resulting procedure is completely decentralized: players need neither knowledge of their opponents' actions nor even a complete understanding of the game itself beyond their own payoffs and strategies.
\end{abstract}

\section{Introduction}
\label{sec:intro}
Algorithmic game theory aims to discover simple, efficient algorithms that converge to equilibria.
While in many cases, such algorithms are useful in their own right, they can also contribute towards more classical uses of game theory such as modelling real-world behavior.
For example, if researchers can show the nonexistence of an algorithm that efficiently computes a given equilibrium, it raises questions on how well that equilibrium could model real-world behavior, since how would the real-world agents themselves discover it?
Conversely, the existence of simple procedures that efficiently converge to an equilibrium increases the plausibility that such an equilibrium could arise naturally.

To model such procedures, we imagine that players play a game repeatedly and choose their actions at each iteration of the game based on a simple learning procedure over their past results. To model selfish agents, we are specifically interested in \emph{uncoupled} learning procedures, where each player's actions do not depend on \emph{another} player's payoffs.
Blackwell's Approachability Theorem (\citeyear{blackwell1956analog}) was later used to show that uncoupled dynamics could lead to a Nash equilibrium for two-player zero-sum games if each player chose their strategies with probabilities proportional to their regret for not choosing that strategy in the past. Separately, \pyear{hannan1957approximation} discovered a similar result via a modified version of fictitious play. More recently, \citeauthor{hart2000simple} (\citeyear{hart2000simple}) extended this result to converge to correlated equilibria for any normal-form game with any number of players. A separate paper also showed the nonexistence of any uncoupled procedure that converged to a Nash equilibrium in the general case \parencite{hart2003uncoupled}.

This paper asks whether we can achieve similar results for extensive-form games, where players play sequentially rather than simultaneously.
One recent successful approach has been counterfactual regret minimization \parencite{zinkevich2008regret}, which converges to the Nash equilibrium for two-player zero-sum extensive-form games.
As such, counterfactual regret minimization has been used extensively in building modern computer poker agents, which have recently surpassed humans in many popular forms of poker \parencite{bowling2015heads, moravvcik2017deepstack, brown2018superhuman, brown2019superhuman}.
Counterfactual regret minimization uses a tree decomposition approach to break up the original game into a large number of ``minigames'', each of which can be independently optimized with any external regret minimization algorithm (e.g. \citeauthor{blackwell1956analog} (\citedate{blackwell1956analog})). We discuss counterfactual regret minimization as well as other methods of decomposing global regrets into local regrets in Section~\ref{sec:globaldecomposition}.

Directly attempting to extend counterfactual regret minimization to converge to the set of correlated equilibria of general extensive-form games runs into the hurdle that describing the set of correlated equilibria may require an exponentially large number of incentive constraints. We consider several possible generalizations of the correlated equilibrium to extensive-form games and focus primarily on a new solution concept we call the forgiving correlated equilibrium.
In a nutshell, the forgiving correlated equilibrium differs from the canonical correlated equilibrium by recommending actions only when a player encounters an infoset instead of revealing the complete recommended strategy at the beginning of the game.
Assuming all other players follow their recommendations, each player is incentivized to follow her recommendations \emph{regardless of whether she has done so at previous infosets}.
We discuss the forgiving correlated equilibrium and its cousins in more detail in Section~\ref{sec:equilibria}.

Our procedure follows in the footsteps of \citeauthor{zinkevich2008regret} (\citeyear{zinkevich2008regret}) by decomposing ``global'' regrets (which concern strategies over the entire game) into ``local'' regrets (which concern only actions at a single infoset). Our approach relies on both the tree decomposition technique that they use as well as a variation of the well known one-shot deviation principle.
The resulting  procedure converges to the set of forgiving correlated equilibria for any finite game of perfect recall.
At each timestep, each player observes her past strategies and payoffs and uses this information to determine her future strategy.
Our procedure is completely decentralized: players need neither knowledge of their opponents' actions nor even a complete understanding of the game itself beyond their own payoffs and strategies.
Previously, \pyear{huang2008computing} give a polynomial-time algorithm converging to the set of extensive-form correlated equilibria, but no algorithm is known to converge to the set of autonomous correlated equilibria (or by extension the set of forgiving correlated equilibria) in polynomial time.

This work was done concurrently and independently from \pyear{celli2020no}.
Their work leverages the recent framework of regret circuits \parencite{farina2018regret} whereas we rely primarily on inspiration from \pyear{zinkevich2008regret} and the one-shot deviation principle, but their resulting procedure shares substantial similarities with the low memory variation of our procedure that converges to extensive-form correlated equilibria (Section~\ref{sec:lowmemory}).

\section{Notation}
\label{sec:notation}
We consider finite extensive-form games of perfect recall.
For the present, we assume that there are no chance moves, though we address this issue separately in Section~\ref{chance}.
Using standard notation for extensive-form games, let $N$ denote the number of players in the game.
Let $\Sigma$ be the set of pure strategy profiles.
Let $\mathcal{I}$ denote the set of all infosets. For any $s \in \Sigma$ and $I \in \mathcal{I}$, let $s(I)$ denote the action the strategy profile $s$ recommends at infoset $I$.
Let the player function $P$ map from an infoset to the player who plays at that infoset. Let the action function $A$ map from an infoset to the set of valid actions at that infoset. In our setting, imagine that players play a game repeatedly. At each timestep $t = 1, 2, \dots$, the players each individually choose their strategies, and we represent their collective choice with the strategy profile $s^t$.

We now introduce the following notation:
Let $O(s, I, a) \in \{0, 1\}$ denote if both the following conditions are true: 
\begin{enumerate}
    \item $s(I) = a$
    \item Assuming everyone plays according to the strategy profile $s$, infoset $I$ is on the path of play.
\end{enumerate}
We can additionally define $O(s, I) = \sum_{a \in A(I)} O(s, I, a)$. Colloquially, we say that an infoset action pair $(I, a)$ is \emph{observed} at timestep t, if $O(s^t, I, a) = 1$.
We may also say that an infoset $I$ is observed at timestep $t$ if $O(s^t, I) = 1$.

Let ($R$ for reachable) $R(s, I) \in \{0, 1\}$ represent whether player $P(I)$ can play in such a way that $I$ is observed, assuming that everyone else plays according to strategy profile $s$.

Let $s | _{I\rightarrow b}$ denote the strategy profile where all players play according to $s$ except that player $P(I)$ plays as necessary to reach infoset $I$ and plays action $b$ at infoset $I$.
Similarly, let $s_I$ be the strategy profile where everyone plays according to $s$ except that player $P(I)$ plays to reach $I$.
If player $P(I)$ cannot play to reach infoset $I$, everyone plays according to $s$ all the time.

Let (N for next) $N(s, I, a)$ represent the infoset where player $P(I)$ next plays after $I$ assuming everyone plays according to $s |_{I\rightarrow a}$.
If player $P(I)$ cannot play to reach infoset $I$ or she does not play again after infoset $I$, we let $N(s, I, a) = \emptyset$.
Let $Succ(I, a) = \bigcup_{s \in \Sigma} \, \{ N(s, I, a) \}$, which intuitively denotes all possible next infosets player $P(I)$ can play at assuming she plays action $a$ at infoset $I$.
Let $DES(I)$ represent all possible infosets which are owned by player $P(I)$ and that are descendants of $I$ (including $I$ itself).

Call the ancestry of an infoset $I$ the sequence of infosets and actions that lead to $I$ from the start of the game from the perspective of player $P(I)$. Specifically, denote it as the longest possible sequence

\[
I_1, a_1, I_2, a_2 \dots I_n, a_n, I_{n+1}
\]

\noindent where for all $k$, $a_k \in A(I_k)$, $P(I) = P(I_k)$, $I_k \in Succ(I_{k-1}, a_{k-1})$ and $I_{n+1} = I$. By the assumption of perfect recall, this ancestry is unique for each infoset. Now, given the above ancestry of an infoset $I$ define the signal history function $S(s, I)$ as the length $n+1$ vector.
\[
\begin{pmatrix} s(I_1) & \dots & s(I_n) & s(I) \end{pmatrix}
\]

Additionally, denote the first $n$ elements of $S(s, I)$ as the partial signal history represented by
\[
S(s, I)_{-1} = \begin{pmatrix} s(I_1) & \dots & s(I_n) \end{pmatrix}
\]

Denote the set of all possible signal histories at an infoset $I$ as $S(I) = \bigcup_{s \in \Sigma} S(s, I)$

When discussing either autonomous correlated equilibria or forgiving correlated equilibria, for any player $i$, define a deviation plan $d$ as a function mapping every infoset $I$ owned by player $i$ and signal history $\ah \in S(I)$ pair to an action $a \in A(I)$.
For a player $i$, call the set of all her deviation plans $\mathcal{D}_i$. 

Let $u(I, s_1, s_2)$ represent the utility player $P(I)$ receives if everyone plays in the following manner.
All players other than $P(I)$ always play according to strategy profile $s_2$.
Player $P(I)$ plays whatever moves are necessary to reach infoset $I$ and then plays according to strategy profile $s_1$ afterward, including at infoset $I$.

Note that because we assume perfect recall, this uniquely determines the strategy of all players as the sequence of actions to reach $I$ is unique.
Additionally, we often abbreviate $u(I, s_2) = u(I, s_2, s_2)$, which intuitively denotes the payoff player $P(I)$ receives when everyone plays according to strategy profile $s_2$ except that $P(I)$ deviates only as necessary to guarantee that she reaches infoset $I$ if it is possible to do so.
If it is not possible for player $P(I)$ to reach infoset $I$ if everyone else plays according to $s_2$, let $u(I, s_1, s_2) = 0$. 
When discussing autonomous correlated equilibria or forgiving correlated equilibria, let $u(I, d, s)$ represent the utility of player $P(I)$ if she plays to reach infoset $I$ and then follows deviation plan $d$ afterwards (including at infoset $I$), while all other players follow strategy profile $s$ all the time. If it is not possible for player $P(I)$ to reach infoset $I$, all players always play according to strategy profile $s$.

\section{Correlating Signals in Extensive-Form Games}
\label{sec:equilibria}
In this section, we define and analyze various equilibria which extend the correlated equilibria solution concept to extensive-form games.
As with the canonical correlated equilibrium, all the following solution concepts have a mediator select a strategy profile $s$ according to some probability distribution $h$ before play begins.
However, unlike the canonical correlated equilibrium, these solution concepts all have the mediator show player $P(I)$ her action recommendation $a \in A(I)$ only when she reaches infoset $I$ instead of revealing the entire strategy recommendation at the beginning of the game.
In a normal-form game, all these solution concepts collapse onto each other and the canonical correlated equilibrium since there is only one infoset to consider.

As a result of how the recommendations are revealed, each of these equilibria (with the exception of forgiving correlated equilibria) forms a superset of the set of correlated equilibria since at any given point in time, each player has only seen a partial strategy recommendation and has less information to decide whether she wants to deviate from the recommendation.

The canonical correlated equilibrium can pose incentive constraints that grow in number exponential to the number of infosets for extensive-form games, and it is currently not known if any algorithm can discover such an equilibrium in polynomial time.
We review three potential candidate solution concepts that generalize the correlated equilibrium to extensive-form games by posing only a polynomial number of (relevant) incentive constraints relative to the number of infosets in the game: the agent form correlated equilibrium, the extensive-form correlated equilibrium, and the autonomous correlated equilibrium. We additionally propose a new solution concept, the forgiving correlated equilibrium, that refines all the above solution concepts.

To ensure standard notation, we give definitions for all equilibria considered in the following sections, as well as intuitive explanations for what they denote.

\subsection{Agent Form Correlated Equilibrium}

The agent form correlated equilibrium was proposed by \citeauthor{forges1986correlated} (\citeyear*{forges1986correlated}). Imagine constructing a variation of an extensive-form game where each infoset is controlled by a separate agent. The correlated equilibrium of that game variant forms an agent correlated equilibrium of the original game.
An alternative equivalent characterization of an agent correlated equilibrium is that there is no profitable one-shot deviation from any action recommendation given at any infoset observed with positive probability.
Of all the generalizations of the correlated equilibrium to extensive-form games we consider, the agent form correlated equilibrium poses the weakest incentive constraints on the players.

\theoremstyle{definition}
\begin{definition}{Agent Form Correlated Equilibrium}

Let $h$ be a probability distribution over strategy profiles. $h$ forms an \emph{agent form correlated equilibrium} if it satisfies:
\[
\max_{I \in \mathcal{I}} \max_{a \in A(I)} \max_{b \in A(I)} \sum_{s \in \Sigma} h(s) O(s, I, a) (u(I, s_{I\rightarrow b}, s) - u(I, s)) \leq 0
\]
\end{definition}

\subsection {Extensive-Form Correlated Equilibrium}
\label{efce}

The extensive-form correlated equilibrium \parencite{von2008extensive} poses a stricter constraint than the agent form correlated equilibrium in that players are not restricted to one-shot deviations. However, upon deviating from the mediator's recommendation, they will no longer receive any future recommendations from the mediator (other players still do).
This differs from the standard correlated equilibrium since in that setting, each player receives her recommendations for the whole game at the beginning of the game before making any of her action choices.

A correlating signal $h$ forms an extensive-form correlated equilibrium if every player is incentivized to follow their signal, assuming that all the other players will follow their signals regardless of her choice.
If $h$ is the signal for an extensive-form correlated equilibrium, then for all infoset action pairs $(I, a)$, at least one of the following two conditions holds

 \begin{enumerate}
   \item Infoset $I$ is never reached if players play according to the recommendation signal $h$ or action recommendation $a$ is never given by $h$ at infoset $I$ (i.e. $\forall s, h(s) O(s, I, a) = 0)$
   \item Conditional on observing infoset $I$ and recommendation $a$, player $P(I)$ would prefer to follow her recommendation and continue receiving information than to unilaterally deviate to any strategy $s'$.
  
 \end{enumerate}
 
\noindent Based on this observation, we can formally define the extensive-form correlated equilibrium.

\theoremstyle{definition}
\begin{definition}{Extensive-Form Correlated Equilibrium}

Let $h$ be a probability distribution over strategy profiles. $h$ forms an \emph{extensive-form correlated equilibrium} if it satisfies:
\[
\max_{I \in \mathcal{I}} \max_{a \in A(I)} \max_{s' \in \Sigma} \sum_{s \in \Sigma} h(s) O(s, I, a) (u(I, s', s) - u(I, s)) \leq 0
\]
\end{definition}

\subsection{Autonomous Correlated Equilibria}

\noindent The autonomous correlated equilibrium (\cite{forges1986approach}) poses an even stricter incentive constraint compared to the extensive-form correlated equilibrium in that signals will continue to be given even after a player deviates from their recommendations. Thus, players will have more information after deviating than their counterparts in an extensive-form correlated equilibrium since the signals received after deviating could reveal information on what the other players are receiving as recommendations. Because we need to account for this additional information, our definition for the autonomous correlated equilibrium (and its refinement, the forgiving correlated equilibrium), use the concept of a deviation plan defined in Section~\ref{sec:notation}.

\begin{definition}
Let $h$ be a probability distribution over strategy profiles. A signal $h$ forms an \emph{autonomous correlated equilibrium} if it satisfies:
\[
\max_{I \in \mathcal{I}} \max_{\ah \in S(I)} \max_{d \in \mathcal{D}_{P(I)}} \sum_{s \in \Sigma} h(s) O(I, s) \mathds{1}(S(I, s) = \ah) (u(I, d, s) - u(I, s)) \leq 0
\]
\end{definition}

To explain the definition using the English language, for every possible infoset and signal history combination ($I, \ah$), summing over all possible strategies weighted by how often the signal gives a specific set of recommendations ($\sum_{s\in \Sigma} h(s)$), for all instances where it is possible to receive action recommendation $\ah$ at infoset $I$ if all players follow their recommendations ($\mathds{1}(S(I, s) = \ah)O(I, s)$), player $P(I)$ has no incentive to deviate from her strategy recommendations at either the current point or any point in the future $(u(I, d, s) - u(I, s))$.

\subsection{Forgiving Correlated Equilibria}

The forgiving correlated equilibrium poses the strictest incentive constraint of all the considered extensions of correlated equilibria to extensive-form games.
As in an autonomous correlated equilibrium, recommendations continue to be given to each player even after they deviate from their recommendations, but a forgiving autonomous correlated equilibrium additionally requires that these recommendations continue to guide the player in the right direction.
As such, in contrast to the autonomous correlated equilibrium and the extensive-form correlated equilibrium, the mediator is ``forgiving'' in that she will continue to give each player the best possible recommendation even after that player has already ignored her advice.

Concretely, we say a signal forms a forgiving correlated equilibrium if, assuming everyone else always follows their signals, you are incentivized to follow your signal, regardless of whether you have done so in previous infosets this game.
Stated another way, deviating from the mediator's recommendation at any point in time will never be profitable, even if you have previously already deviated from the mediator's recommendations.
Note that this is related but not the same as the concepts of subgame perfection and sequential rationality, as we continue to assume that at most one player will ever deviate from her recommendations.

\label{facedefinition}
\begin{definition}
Let $h$ be a probability distribution over strategy profiles. A signal $h$ forms a \emph{forgiving correlated equilibrium} if it satisfies:
\[
\max_{I \in \mathcal{I}} \max_{\ah \in S(I)} \max_{d \in \mathcal{D}_{P(I)}} \sum_{s \in \Sigma} h(s) R(I, s) \mathds{1}(S(I, s) = \ah) (u(I, d, s) - u(I, s)) \leq 0
\]
\end{definition}

Notice the strong similarities between the definitions for forgiving correlated equilibria and autonomous correlated equilibria.
The two definitions, in fact, only differ by one character as the forgiving correlated equilibrium uses the reachability condition ($R$) instead of the observability condition $O$. As such, in the autonomous correlated equilibrium, we constrain only infosets that are on the path of play but in a forgiving correlated equilibrium, we constrain all infosets that are \emph{reachable} under a strategy $s$ with the $R(s, I)$ term assuming only player $P(I)$ deviates from her recommendations. Additionally, note that the forgiving correlated equilibrium refines all the above solution concepts, so any procedure that converges to the set of forgiving correlated equilibria also automatically converges to the set of any of the other equilibria in question.

\subsection{Equilibria in the Real World}

We consider two broad use cases for such equilibria in the real world. One is to model observed behavior, and the other is to advise it.
Extensive-form games are commonly used to model real-world scenarios where actions are taken sequentially instead of all at once.
Like the canonical correlated equilibrium, the autonomous correlated equilibrium (and the forgiving correlated equilibrium by extension) allows for the correlating signal to arise ``naturally’’ for such games.
As a concrete example, we imagine a group of ancient politicians quarantined (separately) in the same general area during a long-forgotten pandemic.
Since there is no internet and nonessential mail is slow during such times, communicating directly with one another is difficult.
Nevertheless, each politician is required to vote on urgent government matters received via priority mail that still does function in such times. 
Every morning, each politician receives a briefing from the government with information on the results of yesterday’s vote and the bill they must vote on today, as well as a collection of private information compiled by their hardworking staff. They must mail in their vote by the end of the day.
Despite the lack of direct communication, politicians can still coordinate if they have access to a common signal. For instance, if they all have windows, the coordinated policy could even be something as simple as all politicians in a coalition voting no if it rains that day.

To return from the analogy to the language of game theory, each infoset is the information each player receives at the start of each day combined with their memory of past briefings, and the exogenous variation of the weather forms the correlating signal. As such, solution concepts such as the forgiving correlated equilibrium have the potential to help explain the behavior of agents playing extensive-form games in the real world.

Yet another purpose can be to advise behavior. 
To a rough approximation, Google Maps can be viewed as a mediator in a forgiving correlated equilibrium.
When we are instructed to take a certain route to get to our destination, it is usually in our interest to follow that route. Additionally, even if we deviate from the recommended route, Google Maps will recalculate the best route given our decision rather than shutting off and leaving us to find our own way to the destination, as the mediator in an extensive-form correlated equilibrium would be apt to do. One can also imagine other applications where such equilibria may prove useful for advising behavior, such as routing Internet traffic or coordinating global shipping.

\section{From Global To Local Regrets}
\label{sec:globaldecomposition}
Having defined the equilibria we hope to converge to, we now proceed to describe our procedure.
Directly attempting to satisfy the incentive constraints describing the various equilibria as defined in Section~\ref{sec:equilibria} is likely challenging. Even simply computing many of the inequalities that define such equilibria by directly enumerating all possible strategy deviations would require exponential amounts of time, because if we let $M = \norm{\max_{I \in \mathcal{I}} \mathcal{A}(I)}$, the number of strategy profiles in the game is on the order of $\norm{\mathcal{I}}^{M}$.
We call such regrets ``global'' if they operate over strategy deviations of the entire game.
To make such computations tractable, one approach, popularized by \pyear{zinkevich2008regret}, is to break down these ``global'' regrets into ``local'' regrets which operate only on single actions at a given infoset. 
If these local regrets can be minimized independently and there are only a polynomial number of (relevant) regrets to minimize, then the runtime can be converted to a much more tractable $O( \texttt{poly}(M \times \norm{\mathcal{I}}))$.

We start by reviewing the tree decomposition approach used by \pyear{zinkevich2008regret} in the counterfactual regret minimization procedure before discussing how we can achieve similar results using a variation of the one-shot deviation principle.

\subsection{Counterfactual Regret Minimization}
\label{cfr}

The following is an overview of counterfactual regret minimization, adapted slightly for our purposes. We walk through the procedure of \pyear{zinkevich2008regret} with the exception that we add an additional $O(s^t, I^P, a)$ term at each step, which allows us to filter for only the timesteps where the pair $(I^P, a)$ was observed.

\subsection{Regrets}

We assume $P(I^P) = P(I)$ and that $I^P$ is an ancestor of $I$.

\theoremstyle{definition}
\begin{definition}{External Regrets ($ER$)}

For any $b \in A(I)$ and $a \in A(I^P)$, we can define the external regret as follows:

\[
ER^T(I^P, a, I) = \max\limits_{s^B \in \Sigma} \frac{1}{T} \sum_{t=1}^T O(s^t, I^P, a) R(s^t, I) (u(I, s^B, s^t) - u(I, s^t))
\]
\[
ER^{T+}(I^P, a, I) = \max (0, ER^T(I^P, a, I))
\]
\end{definition}

Intuitively, $ER$ is the external regret for the subgame rooted at infoset $I$ (from player $P(I)$'s perspective), filtering only for timesteps where the infoset action pair $(I^P, a)$ was observed.
\theoremstyle{definition}
\begin{definition}{Counterfactual Regrets ($CFR$)}

For any $b \in A(I)$ and $a \in A(I^P)$, we can define the $CFR$ (counterfactual regret) as follows:

\[
CFR^{T}(I^P, a, I, b) = \frac{1}{T} \sum_{t=1}^T O(s^t, I^P, a) R(s^t, I) (u(I, s^t |_{I \rightarrow b}, s^t) - u(I, s^t))
\]
\[
CFR^{T+}(I^P, a, I) = \max\limits_{b \in A(I)} \max(0, CFR^{T}(I^P, a, I, b))
\]
\end{definition}

Notice the substantial similarities between $CFR$ and $ER$. While $ER$ regrets are ``global'' in that they allow deviations described by an entire strategy, $CFR$ regrets are ``local'' in that they only allow only changes of a single action. In the following sections, we show how $ER$ can be bounded by a sum of a collection of $CFR$ regrets.

\subsection{Tree Decomposition}

Note that via rearrangement, we have

\begin{align*}
& \max\limits_{s^B \in \Sigma} \: (u(I, s^B, s^t) - u(I, s^t))\\
= \: & \max\limits_{s^B \in \Sigma}  \: (u(I, s^B, s^t) - u(I, s^t) + u(I, s^t |_{I\rightarrow s^B(I)}) - u(I, s^t |_{I\rightarrow s^B(I)})) \\
= \: & \max\limits_{s^B \in \Sigma}\:(u(I, s^t |_{I\rightarrow s^B(I)}) - u(I, s^t)) + (u(I, s^B, s^t) - u(I, s^t |_{I\rightarrow s^B(I)})) \\
\leq \: & \max\limits_{s^B \in \Sigma}\:(u(I, s^t |_{I\rightarrow s^B(I)}) - u(I, s^t)) + \max\limits_{s^B \in \Sigma}\:(u(I, s^B, s^t) - u(I, s^t |_{I\rightarrow s^B(I)}))
\end{align*}

This is the core technique behind counterfactual regret minimization as well as our decomposition in Section~\ref{mainprop}.

\subsection{CFR Inequality}

The main purpose of this inequality is to show that we can minimize our $ER$ regrets by simply minimizing $CFR$ regrets at each infoset that is a descendent of $I^P$.

\begin{prop}
\[
ER^{T}(I^P, a, I) \leq CFR^{T+}(I^P, a, I) + \max\limits_{b \in A(I)} \sum_{I’ \in Succ(I, b)} ER^{T}(I^P, a, I’)
\]
\end{prop}
\begin{proof}

\begin{align*}
ER^{T}(I^P, a, I) = & \; \max\limits_{s^B \in \Sigma} \frac{1}{T} \sum_{t=1}^T O(s^t, I^P, a) R(s^t, I) (u(I, s^B, s^t) - u(I, s^t)) \\
\leq & \; \max\limits_{s^B \in \Sigma}\:\frac{1}{T} \sum_{t=1}^T O(s^t, I^P, a) R(s^t, I) (u(I, s^t |_{I\rightarrow s^B(I)}) - u(I, s^t)) \: + \\
& \;  \max\limits_{s^B \in \Sigma} \frac{1}{T} \sum_{t=1}^T  O(s^t, I^P, a) R(s^t, I)  (u(I, s^B, s^t) - u(I, s^t |_{I\rightarrow s^B(I)}))
\end{align*}

This is an application of the tree decomposition technique discussed previously. Now, examine the two maximization terms separately. The first term is

\begin{align*}
& \max\limits_{s^B \in \Sigma}\:\frac{1}{T} \sum_{t=1}^T O(s^t, I^P, a) R(s^t, I) (u(I, s^t |_{I\rightarrow s^B(I)}) - u(I, s^t)) \\
= & \max\limits_{b \in A(I)}\:\frac{1}{T} \sum_{t=1}^T O(s^t, I^P, a) R(s^t, I) (u(I, s^t |_{I\rightarrow b}) - u(I, s^t)) \\
\leq & \; CFR^{T+}(I^P, a, I)
\end{align*}

The second term is

\begin{align*}
& \max\limits_{s^B \in \Sigma} \frac{1}{T} \sum_{t=1}^T  O(s^t, I^P, a) R(s^t, I)  (u(I, s^B, s^t) - u(I, s^t |_{I\rightarrow s^B(I)})) \\
= & \max\limits_{b \in A(I)} \max\limits_{s^B \in \Sigma} \frac{1}{T} \sum_{t=1}^T  O(s^t, I^P, a) R(s^t, I)  (u(I, s^B|_{I\rightarrow b}, s^t) - u(I, s^t |_{I\rightarrow b})) \\
= &  \max\limits_{b \in A(I)} \max\limits_{s^B \in \Sigma} \sum_{I' \in Succ(I, b)} \frac{1}{T} \sum_{t=1}^T  O(s^t, I^P, a) R(s^t, I')  (u(I', s^B, s^t) - u(I', s^t)) \\
\leq &  \max\limits_{b \in A(I)} \sum_{I' \in Succ(I, b)} \max\limits_{s^B \in \Sigma}  \frac{1}{T} \sum_{t=1}^T  O(s^t, I^P, a) R(s^t, I')  (u(I', s^B, s^t) - u(I', s^t)) \\
= &  \max\limits_{b \in A(I)} \sum_{I' \in Succ(I, b)} ER(I^P, a, I')
\end{align*}

\end{proof}

\begin{corollary}
\[
ER^{T+}(I^P, a, I) \leq CFR^{T+}(I^P, a, I) + \max\limits_{b \in A(I)} \sum_{I’ \in Succ(I, b)} ER^{T+}(I^P, a, I’)
\]
\end{corollary}

We can expand out the above inequality recursively, and write

\begin{prop}

\label{cfrdecomp}
\[
ER^{T+}(I^P, a, I) \leq \sum_{I' \in DES(I)} CFR^{T+}(I^P, a, I')
\]
\end{prop}

\begin{proof}
\begin{align*}
ER^{T+}(I^P, a, I) \leq \; & CFR^{T+}(I^P, a, I) + \max\limits_{b \in A(I^P)} \sum_{I’ \in Succ(I^P, b)} ER^{T+}(I^P, a, I’) \\
\leq \; &  CFR^{T+}(I^P, a, I) +  \sum\limits_{b \in A(I^P)} \sum_{I’ \in Succ(I^P, b)} ER^{T+}(I^P, a, I’) \\
\leq \; &  \sum_{I' \in DES(I)}  CFR^{T+}(I^P, a, I’)
\end{align*}

where the last step expands out $ER$ recursively and uses the assumption of perfect recall to show that each $I'$ appears at most once in the sum.
\end{proof}
In their original paper, \pyear{zinkevich2008regret} minimized (slightly modified) $CFR$ regrets independently at each $I'$ using Blackwell's regret minimization. This resulted in convergence to the set of coarse correlated equilibria, which in the special case of two-player zero-sum games, corresponds to the Nash equilibrium of the game. 

\subsection{Regret Decomposition via the One-Shot Deviation Principle}

An alternative method of decomposing regrets is via a variation of the one-shot deviation principle for subgame perfect Nash equilibria. Proposition~\ref{faceredefinition} asserts that if there is no profitable ``local'' deviation from the signal (even if you have already deviated before), there is no profitable ``global'' deviation from the signal either.

\begin{prop}
\label{faceredefinition}
A signal $h$ satisfies the inequalities given in Definition~\ref{facedefinition} (definition of a forgiving correlated equilibrium) if and only if it satisfies
\[
\max_{I \in \mathcal{I}} \max_{\ah \in S(I)} \max_{b \in A(I)} \sum_{s\in \Sigma} h(s) R(I, s) \mathds{1}(S(I, s) = \ah) (u(I, s_{I\rightarrow b}) - u(I, s_{I})) \leq 0
\]
\end{prop}

\begin{proof}
The backward direction falls out immediately, as all local deviations are also global deviations, so if no global deviations exist, no local ones will either.
For the forward direction, consider the following lemma.
\label{globlocal}
\begin{lemma}
If $h$ is not a forgiving correlated equilibrium (as per Definition~\ref{facedefinition}), there exists at least one pair $(I, \ah)$ such that both of the following are true:

\begin{enumerate}
    \item \textbf{Possibility}. Under signal $h$, it is possible that signal history $\ah$ is observed at infoset $I$ (i.e. $\exists s$ s.t. $h(s)\mathds{1}(S(I, s) = \ah) > 0$)
    \item \textbf{Local Profitability} Player $P(I)$ believes it is profitable to deviate from the signal history $\ah$ at the current infoset, but believes that it will not be profitable to deviate from the recommendation for any descendent infosets $I' \in DES(I)$, $I' \neq I$.
\end{enumerate}
\end{lemma}
\begin{proof}
Consider the set of $(I, \ah)$ pairs that are possible under signal $h$. Define a sequence $L_n$ as follows. Choose $L_1$ to be any possible pair $(I_1, \ah_1)$ where $\ah_1 \in S(I_1)$ and where our player believes it is profitable (in expectation) to deviate from the signal recommendation $\ah_1$ at infoset $I$ and follow some arbitrary deviation plan $d$ from this point onwards. By the assumption that $h$ does not form a forgiving correlated equilibrium, at least one such instance must exist.
For any $n > 1$, choose $L_n$ to be any $(I_n, \ah_n)$ pair that satisfies both the above conditions and also ensures that $I_n \in DES(I_{n-1})$ and $I_n \neq I_{n-1}$.
Because our game is finite and perfect recall guarantees that there are no cycles in our game tree (from the perspective of a single player), this sequence must terminate. The last such $(I_n, \ah_n)$ pair satisfies our criteria.
\end{proof}

We can immediately deduce the following corollary, which is the contrapositive of the lemma.

\begin{corollary}
If there is no such $(I, \ah)$ pair that satisfies both the possibility and local profitability conditions defined above, then there is no profitable global deviation either.
\end{corollary}

This completes our proof, as the definition proposed by Proposition~\ref{faceredefinition} directly encodes the possibility and local profitability conditions for every $I, \ah$ pair.
\end{proof}
\section{Main Procedure}
\label{sec:mainprocedure}
In our setting, we imagine that players play a game repeatedly and choose their actions at each iteration of the game based on a simple learning procedure over their past results. Our goal is to design a procedure that will cause the resulting play to efficiently converge to the desired equilibrium. This is typically done in two steps. First, we define a collection of regrets, chosen such that if all such regrets go to $0$, then the resulting play will converge to the desired equilibrium. Secondly, we show that we can send all such regrets to $0$ by having players play in a particular fashion.
We define the $CFIR$ regrets as follows and show that minimizing these regrets is sufficient to converge to the set of forgiving correlated equilibria.
\begin{definition}{Counterfactual Internal Regrets (CFIR)}
\[
CFIR^{T+}(I, \ah, b) = \sum_{t = 1}^T R(I, s^t) \mathds{1}(S(I, s^t) = \ah) (u(I, s^t_{I\rightarrow b}) - u(I, s^t_{I}))
\]
\[
CFIR^{T+}(I, \ah) = \max_{b \in A(I)} CFIR^{T+}(I, \ah, b)
\]
\end{definition}

\begin{prop}
\label{cfirconvergence}
If $CFIR^{T+}(I, \ah) \rightarrow 0$ for all $I$ and $\ah \in S(I)$ as $T \rightarrow \infty$, then the distance between the signal $h^T(s) = \frac{1}{T}\sum_{t=1}^T \mathds{1}(s^t=s)$ and the set of forgiving correlated equilibria converges to 0.
\end{prop}

\begin{proof}
If this is not the case, then at least one of the inequalities which define the forgiving correlated equilibrium as per Proposition~\ref{faceredefinition} must be violated infinitely often by some $\epsilon > 0$.
Consider the specific $I$ and $\ah$ that define this inequality.
We can write,
\begin{align*}
& CFIR^{T+}(I, \ah) \\
= & \max_{b \in A(I)} \sum_{t = 1}^T R(I, s^t) \mathds{1}(S(I, s^t) = \ah) (u(I, s^t_{I\rightarrow b}) - u(I, s^t_{I})) \\
= & \max_{b \in A(I)} \sum_{s\in \Sigma} h^T(s)  R(I, s^t) \mathds{1}(S(I, s) = \ah) (u(I, s_{I\rightarrow b}) - u(I, s_{I}))
\end{align*}

which is the exact inequality in question. Since we assume that $CFIR^{T+}(I, \ah) \rightarrow 0$, this inequality cannot be violated by at least $\epsilon$ for an infinite number of timesteps.
\end{proof}

\subsection{Procedure}

For simplicity, we use \pyear{hart2000simple} though any other internal regret minimization procedure will also work.
At each timestep, each player traverses through her infosets in any order such that an infoset is encountered only after encountering all of its ancestors. Because perfect recall guarantees that the infosets (for a single player) form a proper tree with no cycles, such an ordering is always possible, and the common breadth-first and depth-first order traversals both satisfy this requirement. If we have our player choose her strategy $s^t(I)$ as we encounter each infoset $I$, this traversal ordering allows us to compute the partial signal history $S(s^t, I)_{-1}$ when we encounter each infoset $I$, since it only requires $s^t$ to be specified for ancestors and not descendants of $I$ (or $I$ itself).

How do have our player choose her actions $s^t(I)$ when she encounters an infoset? Our player recalls from memory the past timesteps where she observed the signal $S(s^t, I)_{-1}$ when encountering infoset $I$ and pretends that she is playing to minimize her internal regrets for only this subset of timesteps. Concretely, we can denote her action at the latest timestep where she observed $S(s^t, I)_{-1}$ at infoset $I$ as $a$. If no such timestep exists, choose $a$ arbitrarily.
We concatenate the partial signal history with this past action to get the full signal history $\ah = (S(s^t, I)_{-1}, a)$.
As per \pyear{hart2000simple}, she either repeats her action $a$ or departs from it to action $b$ with probability proportional to her positive regrets from not deviating to any action $b$ for only this subset of timesteps. 

\subsection{Proof of Convergence}

We say that a sequence of random variables $(L_n)_n$ converges to $L$ with high probability if, for every $\epsilon, \alpha > 0$, we can find $N$ such that $n > N$ implies $\Pr(\norm{L_n - L} > \epsilon) < \alpha$.
To guarantee that our procedure converges to the set of forgiving correlated equilibria, Theorem~\ref{cfirconvergence} shows that we merely need to show that all CFIR regrets go to zero with high probability.

\begin{prop}
\label{highprobabilityconvergence}
If all players play according to the above procedure, $CFIR^{T+}(I, \ah)$ converges to $0$ with high probability for all $I$ and $\ah \in S(I)$.
\end{prop}

For any fixed infoset $I$ and signal history $\ah$, examine the timesteps $t$ where $I$ is reachable and $\ah$ is observable. Mathematically, we can represent this as the set of timesteps $t$ that satisfy
\[
R(I, s^{t}) \mathds{1}(S(I, s^{t}) = \ah) = 1
\]
Say that we visit $(I, \ah)$ at timestep $t$ if $t$ satisfies this condition. For a given $(I, \ah)$, notice that any timestep $t$ where we do not visit $(I, \ah)$ automatically contributes $0$ regret to $CFIR^{T+}(I, \ah)$ due to the definition of $CFIR$. Consider for the sake of argument that we visit $(I, \ah)$ every timestep. If so, then this sequence of $CFIR$ is exactly an internal regret and because our choice of actions is constrained by an internal regret minimizing procedure, $CFIR^{T+}(I, \ah, b)$ converges to 0 with high probability as desired.
Of course, in the general case, we may not visit $(I, \ah)$ at every timestep $t$.
Let $B$ represent the maximum possible regret for any strategy profile choice in our game.
To prove convergence to 0 with high probability in general, choose any $\epsilon < B$ and $\alpha > 0$. Choose $N'$ such that if every timestep visited $(I, \ah)$, $n > N'$ would guarantee that

\[
\Pr(CFIR^{T+}(I, \ah, b) \geq \epsilon) \leq \alpha
\]

Choosing $N = \frac{BN'}{\epsilon}$ satisfies our condition without this assumption. We show this in two cases. For any $n > N$, if less than $N'$ timesteps between $1 \dots n$ visit $(I, \ah)$, then our maximum possible regret is bounded by $\frac{BN'}{N} = \epsilon$ so we satisfy the condition. On the other hand, if $N'$ or more of these timesteps observe $(I, \ah)$, then the condition is satisfied directly by the fact that $N > N'$, as all the remaining timesteps that do not visit $(I, \ah)$ contribute no regret.
\subsection{Chance Nodes}
\label{chance}

Up to this point, we have assumed that our games lack any element of randomness. We can modify our game to include chance nodes controlled by a new ``chance'' player whose strategy is fixed. 
As our procedure includes no interaction between players, the resulting correlating signal will not induce any correlation between the chance player and the other player's recommendations. Additionally, since all regret minimization procedures guarantee that the regrets will go to zero regardless of how other players play, we can immediately guarantee convergence for games with chance moves without any changes to our procedure.

\subsection{Memory Usage}

Our procedure requires that each player keep track of all their regrets and latest moves for every $(I, \ah)$ pair. In the worst case, this results in exponentially large memory usage, as the number of possible signal histories can be exponential in the number of infosets in a game if the tree is very unbalanced. Note though that, if we let $M$ denote the number of infosets in the game, running our procedure for $N$ steps will only require at most $O(NM)$ memory, as only one of the exponentially many regrets will be relevant at any given infoset for a given timestep.
As a result, running our algorithm for a number of steps polynomial in $M$ will also only require memory polynomial in $M$.

\section{Low Memory Variation}
\label{sec:lowmemory}

If, however, memory is cause for concern, we can revise our procedure to require only at most polynomial amounts of memory. This modified procedure will instead converge to the set of extensive-form correlated equilibria which poses weaker incentive constraints than the set of forgiving correlated equilibria. Recall that in an extensive-form correlated equilibrium, the mediator will cease to give recommendations to any player who chooses to ignore her recommendation.
As a result, convergence to the set of extensive-form correlated equilibria does not require tracking a separate regret for each signal history, which was the underlying cause of the exponential memory requirement for the main procedure.

\subsection{Regret Definitions}

We start by defining a pair of regrets ($IR$ and $AR$) which, if minimized, guarantee convergence to extensive-form correlated equilibria and agent form correlated equilibria respectively.

\theoremstyle{definition}
\begin{definition}{Internal Regrets ($IR$)}

For $a \in A(I)$, define
\[
IR^{T+}(I, a) = \max\limits_{s' \in \Sigma} \frac{1}{T} \sum_{t=1}^T \mathds{1}(s'(I) \neq a) O(s^t, I, a) (u(I, s', s^t) - u(I, s^t))
\]
\end{definition}
 Intuitively, $IR$ says that upon receiving a recommendation $a$ at infoset $I$, what is your regret for not instead playing an alternative strategy $s'$ where $s'$ is required to \emph{not} play action $a$ at infoset $I$? $IR$ is always non-negative since in the worst case, a player can revert to her recommended action $a$ and receive 0 regret.
We relate $IR$ to the extensive-form correlated equilibrium with the following proposition.
\begin{prop}
If 
\[
\max\limits_{I \in \mathcal{I}} \max\limits_{a \in A(I)} IR^{T+}(I, a) \rightarrow 0
\]
as $T \rightarrow \infty$ then the distance between the signal $h^T(s) = \frac{1}{T}\sum_{t=1}^T \mathds{1}(s^t=s)$ and the set of extensive-form correlated equilibria converges to 0.
\end{prop}

\begin{proof}

If this is not the case, then at least one of the inequalities which define the set of extensive-form correlated equilibria must be violated infinitely often by some $\epsilon > 0$.
This means that we can find some $I$, $a$, $s'$ such that for an infinite number of timesteps $t$

\[
\sum_{s \in \Sigma} h^t(s) O(s, I, a) (u(I, s', s) - u(I, s)) > \epsilon
\]

By yet another variation of the one-shot deviation principle, we can choose some triplet $(I, a, s')$ that satisfies both the above condition and that $s'(I) \neq a$, as the nonexistence of such a triplet would imply that the inequality would never be violated. We can now write
\begin{align*}
& IR^{T+}(I, a) \\
= & \max\limits_{s' \in \Sigma} \frac{1}{T} \sum_{t=1}^T \mathds{1}(s'(I) \neq a) O(s^t, I, a) (u(I, s', s^t) - u(I, s^t)) \\
= & \max\limits_{s' \in \Sigma} \sum_{s \in \Sigma} h^T(s) \mathds{1}(s'(I) \neq a) O(s, I, a) (u(I, s', s) - u(I, s)) \\
= & \max\limits_{s' \in \Sigma} \sum_{s \in \Sigma} h^T(s) O(s, I, a) (u(I, s', s) - u(I, s)) \rightarrow 0 \\
\end{align*}

This contradicts our condition for selecting $(I, a, s')$, thus proving the proposition.
\end{proof}

\theoremstyle{definition}
\begin{definition}{Agent Form Correlated Equilibrium Regrets ($AR$)}

For $a, b \in A(I)$ define
\[
AR^{T+}(I, a, b) = \frac{1}{T} \sum_{t=1}^T  O(s^t, I, a) (u(I, s^t|_{I\rightarrow b}, s^t) - u(I, s^t))
\]
\[
AR^{T+}(I, a) = \max_{b \in A(I)} AR^{T+}(I, a, b)
\]
\end{definition}

$AR$ is very similar to $IR$, except that only deviations of a single action are allowed. Note that AR can never be negative since in the worst case, player $P(I)$ can just set her action equal to the recommendation of $a$. We relate $AR$ to the agent form correlated equilibrium with the following proposition.
\begin{prop}
If 

\[
\max\limits_{I \in \mathcal{I}} \max\limits_{a \in A(I)} AR^{T+}(I, a) \rightarrow 0
\]
as $T \rightarrow \infty$ then the distance between the signal $h^T(s) = \frac{1}{T}\sum_{t=1}^T \mathds{1}(s^t=s)$ and the set of agent correlated equilibria converges to 0.
\end{prop}

\begin{proof}

If this is not the case, then at least one of the inequalities which define the set of agent form correlated equilibria must be violated infinitely often by some $\epsilon > 0$.
Consider the specific $I$ and $a$ that define this inequality.
We can write
\begin{align*}
& AR^{T+}(I, a) \\
= & \max_{b \in A(I)} \frac{1}{T} \sum_{t=1}^T  O(s^t, I, a) (u(I, s^t|_{I\rightarrow b}, s^t) - u(I, s^t)) \\
= & \max_{b \in A(I)} \frac{1}{T} \sum_{s \in \Sigma}  h^T(s) O(s, I, a) (u(I, s|_{I\rightarrow b}, s) - u(I, s)) \\
\end{align*}

which is the exact inequality in question. Since we assume that $AR^{T+}(I, a) \rightarrow 0$, this is not possible.
\end{proof}
\subsection{Regret Decomposition}
\label{mainprop}

We can leverage the tree decomposition technique again to decompose our global internal regret into a sum of local $AR$ and $CFR$ regrets. This approach is very similar to the approach for counterfactual regret minimization described in Section~\ref{cfr}.

\begin{prop}
\label{irinequality}
\[
IR^{T+}(I^P, a) \leq AR^{T+}(I^P, a) + \max\limits_{b \in A(I^P)} \mathds{1}(b \neq a) \sum_{I \in Succ(I^P, b)} ER^{T}(I^P, a, I)
\]
\end{prop}
\begin{proof}

\begin{align*}
IR^{T+}(I^P, a) =& \; \max\limits_{s^B \in \Sigma} \frac{1}{T} \sum_{t=1}^T \mathds{1}(s^B(I^P) \neq a) O(s^t, I^P, a) (u(I^P, s^B, s^t) - u(I^P, s^t)) \\
\leq & \; \max\limits_{s^B \in \Sigma}\:\frac{1}{T} \sum_{t=1}^T \mathds{1}(s^B(I^P) \neq a) O(s^t, I^P, a)  (u(I^P, s^t |_{I^P\rightarrow s^B(I^P)}) - u(I^P, s^t)) \: + \\
& \;  \max\limits_{s^B \in \Sigma} \frac{1}{T} \sum_{t=1}^T  \mathds{1}(s^B(I^P) \neq a) O(s^t, I^P, a)  (u(I^P, s^B, s^t) - u(I^P, s^t |_{I^P\rightarrow s^B(I^P)}))
\end{align*}

This is an application of the tree decomposition technique discussed previously. Now, examine the two maximization terms separately. The first term is

\begin{align*}
& \max\limits_{s^B \in \Sigma}\:\frac{1}{T} \sum_{t=1}^T \mathds{1}(s^B(I^P) \neq a) O(s^t, I^P, a)  (u(I^P, s^t |_{I^P\rightarrow s^B(I^P)}) - u(I^P, s^t)) \\
= & \max\limits_{b \in A(I^P)}\:\frac{1}{T} \sum_{t=1}^T \mathds{1}(b \neq a) O(s^t, I^P, a)  (u(I^P, s^t |_{I^P\rightarrow b}) - u(I^P, s^t)) \\
= & \max\limits_{b \in A(I^P)}\:\frac{1}{T} \sum_{t=1}^T O(s^t, I^P, a)  (u(I^P, s^t |_{I^P\rightarrow b}) - u(I^P, s^t)) \\
= & \: AR^{T+}(I^P, a)
\end{align*}

The equality on the third line is given by the fact that when $a=b$ the AR regret is 0. For the second term,

\begin{align*}
&  \max\limits_{s^B \in \Sigma} \frac{1}{T} \sum_{t=1}^T  \mathds{1}(s^B(I^P) \neq a) O(s^t, I^P, a)  (u(I^P, s^B, s^t) - u(I^P, s^t |_{I^P\rightarrow s^B(I^P)})) \\
= & \max\limits_{b \in A(I^P)} \max\limits_{s^B \in \Sigma} \frac{1}{T} \sum_{t=1}^T  \mathds{1}(b \neq a) O(s^t, I^P, a)  (u(I^P, s^B|_{I^P\rightarrow b}, s^t) - u(I^P, s^t |_{I^P\rightarrow b})) \\
= & \max\limits_{b \in A(I^P)} \mathds{1}(b \neq a) \max\limits_{s^B \in \Sigma}  \sum_{I \in Succ(I^P, b)} \frac{1}{T}  \sum_{t=1}^T  O(s^t, I^P, a) R(s^t, I) (u(I, s^B, s^t) - u(I, s^t)) \\
\leq & \max\limits_{b \in A(I^P)} \mathds{1}(b \neq a)   \sum_{I \in Succ(I^P, b)} \max\limits_{s^B \in \Sigma} \frac{1}{T}  \sum_{t=1}^T  O(s^t, I^P, a) R(s^t, I) (u(I, s^B, s^t) - u(I, s^t)) \\
= & \max\limits_{b \in A(I^P)} \mathds{1}(b \neq a) \sum_{I \in Succ(I^P, b)} ER^{T}(I^P, a, I)
\end{align*}
\end{proof}
This inequality can be interpreted as follows. Consider a player who is deciding what actions to take after receiving an action recommendation $a$ at infoset $I$.
Assuming she would like to deviate from the recommendation, we can split her decision into two choices: first, she must decide what action to play immediately at infoset $I$.
Secondly since she will receive no further information, she must decide on a strategy for all future infosets.
Notice that the first choice is an ``internal'' regret since the player is given an action recommendation at that infoset, while the second choice is an ``external'' regret since the player receives no additional information upon reaching any future infosets. The first choice's regret corresponds exactly to $AR$, while the second corresponds exactly to $ER$.

Combining Proposition~\ref{irinequality} and Proposition~\ref{cfrdecomp}, we can write \textbf{Equation~\ref{localirdecomp}}:
\label{localirdecomp}
\[
IR^{T+}(I^P, a) \leq AR^{T+}(I^P, a) + \max\limits_{b \in A(I^P)} \mathds{1}(b \neq a) \sum_{I \in Succ(I^P, b)} \sum_{I' \in DES(I)} CFR^{T+}(I^P, a, I')
\]

Equation~\ref{localirdecomp} upper bounds the overall internal regret with the combination of both local internal (AR) and external (CFR) regrets.
\subsection{Low Memory Procedure}
\textbf{Part 1.} At each timestep $t$, one iteration of the game is played as follows, starting from the root: when an infoset $I$ is observed, denote the action player $P(I)$ played the last time she observed infoset $I$ as $a$.
Our player recalls from memory all previous timesteps where she received action recommendation $a$ at infoset $I$ and chooses her next move to minimize her internal regret determined by examining only such timesteps. Doing this step and ignoring the remainder of the algorithm (e.g. by playing arbitrarily) converges to an agent form correlated equilibrium.

\textbf{Part 2.} For each remaining infoset $I$ off the path of play at timestep t, player $P(I)$ chooses her actions as follows: she finds the closest ancestor infoset $I^P$ that is both owned by her and on the path of play at timestep t. Let $a$ denote the action she played at the current timestep at infoset $I^P$. At infoset $I$, she recalls all previous timesteps where she played $a$ at $I^P$ and $I$ was off the path of play. She now plays to minimize her external regret determined by examining only this subset of timesteps. If no such $I^P$ exists, she plays arbitrarily.

\subsection{Intuition}

In a slightly more informal style, but perhaps a bit more intuitively, we describe the process again using the specific internal regret minimization procedure from \pyear{hart2000simple} and the external regret minimization procedure from \pyear{blackwell1956analog}.

We imagine our player is playing a game repeatedly and must now choose her strategy for the current iteration of the game.
When it is her turn to play at any infoset $I^P$, she recalls the last time $I^P$ was on the path of play and her action $a^p$ for that infoset from memory.
As per \pyear{hart2000simple}, she can either repeat $a^p$ or deviate to any other move $b$ with probability proportional to her (positive) regret for not playing $b$ in the past in this exact situation.
Regardless, denote her final choice as $a$. By the assumption of perfect recall, after committing to action $a$, there are now several infosets belonging to our player that can no longer ever be on the path of play for the current iteration of the game, regardless of future play.
Our player may wonder what would have happened if she had instead chosen differently and somehow eventually stumbled upon one of those infosets; call it a fear of missing out. How would she have wanted to play?

For any such infoset I, we have her play action b with probability proportional to her regret described by $CFR^{T+}(I^P, a, I, b)$, as per \pyear{blackwell1956analog}.
In other words, filtering only for iterations where she eliminated the chance to see infoset $I$ by playing action $a$ at $I^P$, if she had instead not chosen to play $a$ and somehow instead stumble upon infoset $I$, to play action $b$ at infoset $I$ proportional for her regret for not playing action $b$ in this exact past situation.

In terms of why this play converges to an extensive-form correlated equilibrium, playing by a regret minimization algorithm sends all the relevant CFR regrets (and thus by extension ER regrets) and AR regrets to zero. ER regrets going to zero shows that had she deviated from the recommendation, her strategy post deviation is optimal. AR regrets going to zero show that conditioned on her strategies post deviation, following the recommendations is optimal. The combination of both of these regrets going to zero thus suggests that there is no profitable deviation of any kind from the recommendation, thus forming an extensive-form correlated equilibrium.

\subsection{Convergence Results}
\label{aceprop}
\begin{prop}
If all players play according to part 1 of the above procedure, then the distance between the signal $h(s) = \frac{1}{T}\sum_{t=1}^T \mathds{1}(s^t=s)$ and the set of agent correlated equilibria converges to 0 as $T \rightarrow \infty$ with high probability.
\end{prop}

\begin{proof}

By the definition of an agent form correlated equilibrium, we merely need to show that $AR^{T+}(I, a) \rightarrow 0$ for all $I, a$ pairs as $T \rightarrow \infty$ with high probability. We proceed similarly to our proof for Proposition~\ref{highprobabilityconvergence}.
Consider any $(I, a)$ pair. Say that an $(I, a)$ pair is observed at timestep $t$ if $O(s^t, I, a) = 1$. Timesteps where $(I, a)$ is not observed contribute no regret to $AR$ as per the definition. Additionally, notice that, if we let $X(T)$ represent only the subset of timesteps from $1 \dots T$ where $(I, a)$ is observed, we can write

\begin{align*}
AR^{T+}(I, a) & = \max_{b \in A(I)} \frac{1}{T} \sum_{t=1}^T  O(s^t, I, a) (u(I, s^t|_{I\rightarrow b}, s^t) - u(I, s^t)) \\
& = \max_{b \in A(I)} \frac{1}{T} \sum_{t \in X(T)} (u(I, s^t|_{I\rightarrow b}, s^t) - u(I, s^t))
\end{align*}

The last statement forms the internal regret which our procedure chooses actions to minimize, so we can infer immediately that the regrets converge to 0 with high probability if $(I, a)$ is observed at every timestep.
Let $B$ denote the maximum possible $AR$ regret. To prove convergence with high probability without that assumption, fix any $\epsilon < B$ and $\alpha > 0$. Let $N'$ be chosen to satisfy convergence with high probability condition under the assumption that $(I, a)$ is observed at every timestep. Choosing $N = \frac{BN'}{\epsilon}$ guarantees that any $n > N$ has

\[
\Pr(AR^{n+}(I, a) \geq \epsilon) \leq \alpha
\]

We show this in two cases. For any $n > N$, if less than $N'$ timesteps between $1 \dots n$ observe $(I, a)$, then our maximum possible regret is bounded by $\frac{BN'}{N} = \epsilon$ so we satisfy the condition. On the other hand, if $N'$ or more of these timesteps observe $(I, a)$, then the condition is satisfied directly by the fact that $N > N'$, as all the remaining timesteps that do not visit $(I, a)$ contribute no regret.
\end{proof}

\begin{prop}
If all players play according to the entire above procedure, then the distance between the signal $h(s) = \frac{1}{T}\sum_{t=1}^T \mathds{1}(s^t=s)$ and the set of extensive-form correlated equilibria converges to 0 as $T \rightarrow \infty$ with high probability.
\end{prop}

\begin{proof}

Proposition~\ref{aceprop} shows that following just part 1 of the main procedure and playing arbitrarily for the remaining infosets sends the $AR$ regrets to $0$ with high probability.
Using Equation~\ref{localirdecomp}, all that remains to be shown is that part 2 of the procedure sends all the relevant $CFR$ regrets to $0$ with high probability.

Call a triplet $(I^P, a, I)$ valid if $P(I) = P(I^P)$ and there exists no strategy profile $s$ such that both of the following are true.

\begin{enumerate}
    
    \item $I$ and $I^P$ are both on the path of play according to $s$
    \item $s(I^P) = a$
\end{enumerate}
Equation~\ref{localirdecomp} shows that for $CFR^{T+}(I^P, a, I)$ to contribute to $IR^{T+}(I^P, a)$, it is a necessary but not sufficient condition for the triplet $(I^P, a, I)$ to be valid.
Say that we visit a valid triplet $(I^P, a, I)$ at timestep $t$ if $O(s^t, I^P, a)R(s^t, I) = 1$. Timesteps where $(I^P, a, I)$ is not visited contribute no regret to $CFR$ as per the definition. 

Now, consider the following lemma.

\begin{lemma}
For any given timestep $t$ and infoset $I$, there exists at most one pair $(I^P, a)$ such that both of the following conditions are true.

\begin{enumerate}
    \item $(I^P, a, I)$ forms a valid triplet
    \item Timestep $t$ contributes nonzero regret to $CFR(I^P, a, I)$
\end{enumerate}
\end{lemma}

\begin{proof}
For a given $I$, consider all possible valid triplets $(I^P, a, I)$. For a timestep $t$ to contribute nonzero regret to $CFR(I^P, a, I)$, the $O(s^t, I^P, a)$ term in the definition of CFR requires $I^P$ to be on the path of play at timestep $t$. Due to assumptions of perfect recall, there exists a single line of ancestors (from the perspective of player $P(I)$ from the root to $I$).
We claim that the unique infoset action pair (if it exists) is given by $(I^P, s^t(I^P))$ where $I^P$ is selected to be the most recent common ancestor of $I$ that is on the path of play.
By construction, any descendants of $I^P$ will not be on the path of play (and thus never satisfy the $O(s^t, I^P, a) = 1$ condition for any a). Any ancestor of $I^P$ (call it $I'$) will satisfy the observability condition with the pair $(I', s^t(I'))$ but will not be a valid triplet since $s^t(I')$ is the action that leads to $I$ (by perfect recall) and thus $(I', s^t(I'), I)$ will not be a valid triplet. Finally, if $a$ is selected to be an action other than $s^t(I')$ then the observability condition will again be violated since by construction $O(s^t, I', a) = 0$.
\end{proof}

A direct application of the above lemma shows that each timestep is responsible for minimizing a single $CFR(I^P, a, I)$ for each infoset $I$, meaning there are no conflicts between the regrets being minimized. Additionally, notice that, if we let $X(T)$ represent only the subset of timesteps where $(I^P, a, I)$ is observed, we can write
\begin{align*}
CFR^{T+}(I^P, a, I) & = \max\limits_{b \in A(I)} \frac{1}{T} \sum_{t=1}^T O(s^t, I^P, a) R(s^t, I) (u(I, s^t |_{I \rightarrow b}, s^t) - u(I, s^t)) \\
& = \max\limits_{b \in A(I)} \frac{1}{T} \sum_{t \in X(T, I^P, a, I)} (u(I, s^t |_{I \rightarrow b}, s^t) - u(I, s^t)) 
\end{align*}
The last statement forms the external regret which our procedure directly minimizes, so we immediately can infer that the regrets converge to 0 with high probability if $(I^P, a, I)$ is observed at every timestep.
Let $B$ denote the maximum possible $CFR$ regret.
To prove convergence with high probability for the overall sequence, fix any $\epsilon < B$ and $\alpha > 0$.
Let $N'$ be chosen to satisfy convergence with high probability condition under the assumption that every $(I^P, a, I)$ is observed at every timestep. Choosing $N = \frac{BN'}{\epsilon}$ guarantees that any $n > N$ has

\[
\Pr(CFR^{T+}(I^P, a, I) \geq \epsilon) \leq \alpha
\]

\noindent We show this in two cases. For any $n > N$, if less than $N'$ timesteps between $1 \dots n$ observe $(I^P, a, I)$, then our maximum possible regret is given by $\frac{BN'}{N} = \epsilon$ so we are guaranteed to satisfy the condition with probability $1$. On the other hand, if $N'$ or more of these timesteps observe $(I^P, a, I)$, then the condition is satisfied directly by the fact that $N > N'$, as all remaining timesteps that do not visit $(I^P, a, I)$ contribute 0 regret.
\end{proof}

\section{Acknowledgements}
I am tremendously grateful to my advisor Gabriel Carroll for his support and guidance throughout the process of writing this thesis, as well as Marcelo Clerici-Arias for making this possible by directing the Economics Honors program. I would also like to thank Daniel Duckworth, Jeffrey Chang, Felix Wang, and Xiaoyu He for their helpful suggestions and comments, and my ECON 180 classmates for helping me learn game theory.

\clearpage
\xpatchbibmacro{date+extradate}{%
  \printtext[parens]%
}{%
  \setunit{\addperiod\space}%
  \printtext%
}{}{}
\printbibliography

\end{document}